\documentclass[journal,comsoc]{IEEEtran}
\usepackage[T1]{fontenc}

\usepackage{mathtools, cuted}
\usepackage{lipsum, color}
\usepackage{amsthm}
\newtheorem{theorem}{Theorem}
\newtheorem{lemma}{Lemma}
\usepackage{mathtools}
\usepackage{graphics}
\DeclarePairedDelimiter\norm{\lVert}{\rVert}
\usepackage{amsmath, amsfonts}
\usepackage{nopageno}
\usepackage{algorithm}
\usepackage{algpseudocode}
\usepackage{float}
\usepackage{verbatim}
\usepackage{cite}

\ifCLASSINFOpdf
\else
\fi
\usepackage{amsmath}

\usepackage[T1]{fontenc}
\usepackage[cmintegrals]{newtxmath}

\begin{document}

\title{Non-Orthogonal Multiple Access Based on Hybrid Beamforming for mmWave Systems}

%
%

\author{
    \IEEEauthorblockN{Mojtaba Ahmadi Almasi\IEEEauthorrefmark{1}, Hani Mehrpouyan\IEEEauthorrefmark{1}}\\
    \IEEEauthorblockA{\IEEEauthorrefmark{1}{\small{Department of Electrical and Computer Engineering, Boise State University,
   \{mojtabaahmadialm,hanimehrpouyan\}@boisestate.edu}}}\\
    

\vspace{-15pt}
}
\maketitle
	\thispagestyle{empty}

\begin{abstract}
This paper aims to study the utilization of non-orthogonal multiple access (NOMA) in hybrid beamforming (HB) multi-user systems called HB-NOMA to serve a large number of mobile users (MUs). First, a sum-rate expression for the HB-NOMA problem is formulated. Second, a suboptimal algorithm is proposed to maximize the sum-rate. Then, a lower bound for the achievable rate is derived under  the condition that the angle between the effective channel vectors of the MU with the highest channel gain and other MUs located inside a cluster is non-zero, which we denote by imperfect correlation. The lower bound indicates that an inefficient MU cluster, can cause severe inter-cluster interference in the network. To verify our findings, numerical simulations have been conducted. 
\end{abstract}


%
\IEEEpeerreviewmaketitle

\section{Introduction}
Fifth generation (5G) wireless networks are deemed to be a promising technology for the rapid growth of mobile data traffic demand. Millimeter-Wave (mmWave) communications operating in the $28$-$300$ GHz range has been proposed as one of the feasible solutions for 5G networks~\cite{r1}. On one hand, the existence of a large communication
bandwidth at mmWave frequencies represents the potential for
significant throughput gains. On the other hand, the shorter wavelengths at the mmWave band allow for the deployment of a large number of antenna elements in a small area, which enables
mmWave systems to potentially support a higher degree
of beamforming gain and multiplexing~\cite{r1}. However, significant path loss, blockage, and hardware limitations are major obstacles for the deployment of mmWave systems. To address these obstacles, several mmWave systems
have been proposed to date~\cite{r5,el2014spatially,r9,r19,almasi2018new}.

An analog beamforming mmWave system is designed in~\cite{r5} which uses one radio frequency (RF) chain and can support only one data stream. In order to transmit multiple streams, by exploiting several RF chains hybrid beamforming (HB) mmWave systems are designed~\cite{el2014spatially}. In~\cite{r9}, the concept of beamspace multi-input multi-output (MIMO) is introduced where several RF chains are connected to a lens antenna array via switches. Finally, multi-beam lens-based reconfigurable antenna MIMO systems are newly proposed to overcome severe pathloss and shadowing in mmWave frequencies~\cite{r19,almasi2018new}. Beside mmWave communications, non-orthogonal multiple access (NOMA) is another enabling technique for 5G networks. NOMA can be used to augment the spectral efficiency in multi-user scenarios~\cite{higuchi2015non}. 

In the aforementioned systems, each beam is used to serve only one mobile user (MU)~\cite{alkhateeb2015limited,almasi2018reconfigurable}. The work in~\cite{alkhateeb2015limited}, shows that exploiting hybrid beamforming in multi-user systems achieves higher spectral efficiency. Also,~\cite{almasi2018reconfigurable} enhances the spectral efficiency by supporting several MUs through multi-beam reconfigurable antenna. Nevertheless, the number of served MUs maybe not sufficient to support the larger user base in 5G networks. To overcome this issue, the integration of NOMA in mmWave systems (mmWave-NOMA) which allows several MUs to share the same beam has been studies in~\cite{chen2017exploiting,ding2017noma,wang2017spectrum,hao2017energy,yang2017noma,wu2017non}. In~\cite{chen2017exploiting,ding2017noma,wang2017spectrum,hao2017energy,yang2017noma} NOMA is evaluated for mmWave systems assuming only baseband precoders/combiners, whereas the work in~\cite{wu2017non} has recently studied NOMA in hybrid beamforming systems. In particular, due to use of hybrid beamforming, the digital precoder of the base station (BS) is not perfectly aligned with the MU's effective channel. While considering this imperfect alignment, a power allocation algorithm that maximizes the sum-rate has been proposed. However, the approach in~\cite{wu2017non} does not study the effect of imperfect beamforming on the sum-rate. Indeed, the approach in~\cite{wu2017non} can be considered a preliminary study applicable to small scale systems, since only two clusters with each of them containing only two MUs are considered. 

In this paper, inspired by the hybrid beamforming system in~\cite{alkhateeb2015limited} and the work in~\cite{wu2017non}, we study the impact of using NOMA on the performance of hybrid beamforming systems which we denote as HB-NOMA. To this end, we consider a scenario in which a single BS is equipped with a hybrid beaforming system similar to that of~\cite{alkhateeb2015limited}. The BS transmits several streams through the formed beams and each beam is allowed to be shared with multiple MUs inside a cluster. Due to application of analog components, unlike digital precoders, in hybrid beamforming, the precoders are not able to perfectly align a beam toward the effective channel of all MUs. This causes the angle between the effective channel vectors of the first MU (see Fig. 1) and other MUs located in the same cluster to be non-zero. Therefore, this paper evaluates the effect of the imperfect alignment on the achievable rate of MUs. The contributions of this paper can be summarized as
\begin{enumerate}
    \item We propose a generalized hybrid beamforming-based NOMA system that consists of one BS and $N$ clusters, where each cluster has $M$ MUs (see Fig. 1). It is assumed that the BS performs hybrid beamforming, i.e., analog/digital precoder, whereas the MUs only perform analog combining. The sum-rate is formulated for the $m$th MU in the $n$th cluster.
    \item A suboptimal algorithm is proposed in three steps to maximize the sum-rate. In  the first step, we design the analog precoder/combiner while in the second step, we derive the digital precoder architecture. Finally, we formulate a suboptimal power allocation technique for the proposed setup. 
    \item We derive a lower bound for the achievable rate of the $m$th MU in the $n$th cluster in the existence of imperfect correlation amongst the channels of MUs in the same cluster. Our analysis shows that under the assumption of imperfect correlation, there can be a significant data-rate loss. 
    \end{enumerate}

The paper is organized as follows: Section~\ref{sec:system} presents the
system model for the HB-NOMA system. In Section~\ref{sec:problem}, the sum-rate problem for the HB-NOMA system is formulated. In Section~\ref{algorithm}, the proposed algorithm for maximizing the sum-rate is presented. Section~\ref{sec:lower} derives a lower bound for the achievable rate of an HB-NOMA user. In Section~\ref{sec:simulation}, we present simulations investigating the performance of the rate. Section~\ref{sec:conclusion} concludes the paper.

\textbf{Notations:} Hereafter, $j = \sqrt{-1}$, small letters, bold letters and bold capital letters will designate scalars, vectors, and matrices, respectively. Superscripts $(\cdot)^{\dagger}$ and $(\cdot)^{*}$ denote the transpose  and transpose-conjugate operators, respectively. Further, $|\cdot|$, and $\norm[]{\cdot}^2$ denote the absolute value, and norm-$2$ of $(\cdot),$ respectively. Finally, the element in $i$th  row of $j$th column of matrix $\mathbf{X}$ is denoted by $\mathbf{X}(i,j)$.

\section{System Model}\label{sec:system}
We consider a mmWave downlink system for 5G wireless networks composed of a BS and multiple MUs as shown in Fig.~\ref{fig:system}. It is assumed that the BS is equipped with  $N_\text{RF}$ chain and $T_\text{BS}$ antennas, while each MU has one RF chain and $T_\text{MU}$ antennas. Further, we assume that the BS communicates with each MU via only one stream, which is in line with prior work such as~\cite{alkhateeb2015limited}. Note that in traditional hybrid beamforming based multi-user systems it is assumed that the maximum number of MUs that can be simultaneously served by the BS equals to the number of BS RF chains~\cite{alkhateeb2015limited}. 

In order to establish better connectivity in a dense area and further improve the sum-rate, this paper utilizes NOMA in hybrid beamforming multi-user systems which is denoted by HB-NOMA. To this end, each beam is allowed to serve more than one MU. For the sake of simplicity, the number of MUs served by each beam is identical and equal to $M$. Also, the MUs are grouped into $N\leq N_\text{RF}$ clusters. Hence, the proposed HB-NOMA system can simultaneously serve $MN \gg N_\text{RF}$ MUs. 

On the downlink, the hybrid beamforming is done through two stages. In the first stage, the transmitter applies an $N\times N$ baseband precoder $\mathbf{F}_\text{BB}$ using its $N_\text{RF}$ RF chains. Next, using analog phase shifters a $T_\text{BS} \times N$ RF precoder, $\mathbf{F}_\text{RF}$, is applied. Thus, the transmit signal vector, $\mathbf{x}$, is given by  
\begin{equation} \label{eq1}
    \mathbf{x} = \mathbf{F}_\text{RF}\mathbf{F}_\text{BB}\mathbf{s},
\end{equation}
where $\mathbf{s} = [s_1, s_2, \cdots, s_N]^{\dagger}$ denotes the information signal vector. Each $s_{n} = \sum_{m = 1}^{M}\sqrt{P_{n,m}^\prime}s_{n,m}$ is the superposition coded signal due to NOMA with $P_{n,m}^\prime$ and $s_{n,m}$ denoting the transmit power and transmitted information signal for the $m$th MU in the $n$th cluster, respectively. Hereafter, MU-$(n,m)$ denotes the $m$th user in the $n$th cluster. Since $\mathbf{F}_\text{RF}$ is implemented by using analog phase shifters, it is assumed that all elements of $\mathbf{F}_\text{RF}$ have equal norm, i.e., $\left|\mathbf{F}_\text{RF}(i,j)\right|^2  = T_\text{BS}^{-1}$ for $i=1,2,\cdots,M$ and $j=1,2,\cdots,N$. Further, the total power of the hybrid transmitter is constrained to $\norm[\big]{\mathbf{F}_\text{RF}\mathbf{F}_\text{BB}}^2_F = N$. The received signal at  MU-$(n,m)$, $\mathbf{r}_{n,m}$, is given by
\begin{equation}\label{eq2}
    \mathbf{r}_{n,m} = \mathbf{H}_{n,m}\mathbf{F}_\text{RF}\mathbf{F}_\text{BB}\mathbf{s} + \mathbf{n}_{n,m},
\end{equation}
where $\mathbf{H}_{n,m}$ of size $T_\text{MU}\times T_\text{BS}$ denotes the mmWave channel between the BS and MU-(n,m). $\mathbf{n}_{n,m}\sim \mathcal{N}(\mathbf{0},\sigma^2\mathbf{I})$ is the additive white Gaussian noise vector of size $T_\text{MU}\times 1$.      

\begin{figure}
\vspace*{-0.7cm}
\hspace*{-0.9cm}
    \centering
    \includegraphics[scale=1.175]{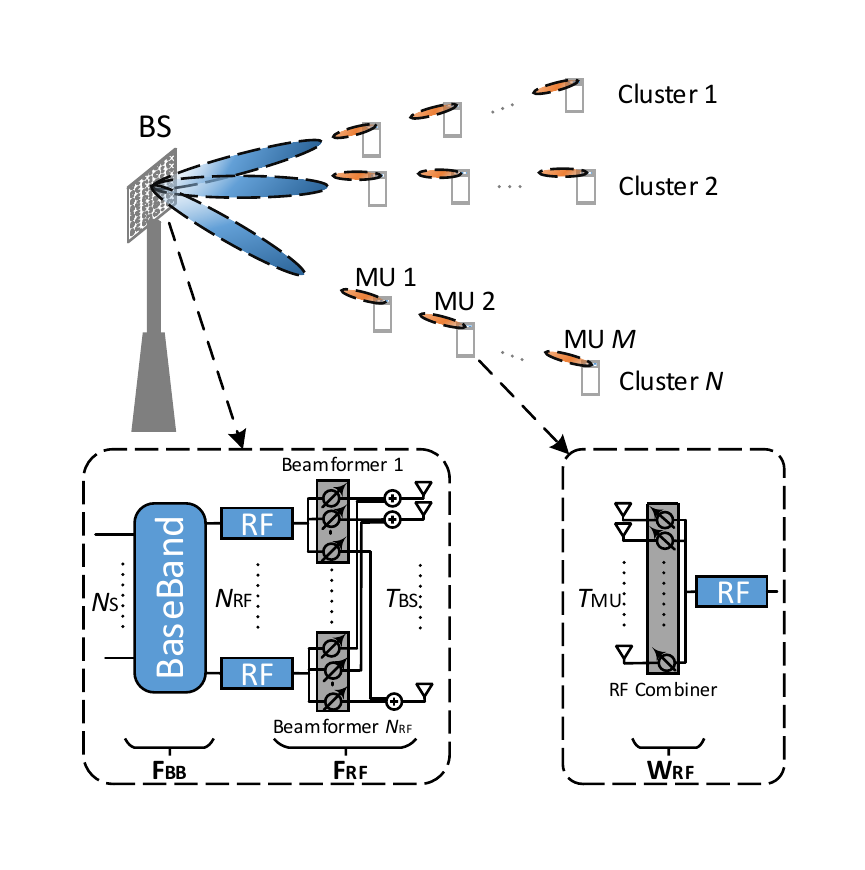}
    \vspace*{-1cm}
    \caption{Schematic of the HB-NOMA system with one BS and \textit{NM} MUs.}
    \label{fig:system}
\end{figure}

At MU-$(n,m)$, the RF combiner, $\mathbf{w}_{n,m}$,  is used to process the received vector as
\begin{align}\label{eq3}
    y_{n,m} &=  \underbrace{\mathbf{w}_{n,m}^*\mathbf{H}_{n,m}\mathbf{F}_\text{RF}\mathbf{f}^n_\text{BB}\sqrt{P_{n,m}^\prime}s_{n,m}}_{\text{desired signal}}   \nonumber \\
    & \ \ +\underbrace{\mathbf{w}_{n,m}^*\mathbf{H}_{n,m}\mathbf{F}_{RF}\mathbf{f}^n_{BB}\sum_{k \neq m}^M\sqrt{P_{n,k}^\prime}s_{n,k}}_{\text{intra-cluster interference}}  \nonumber \\
    & \ \ + \underbrace{\mathbf{w}_{n,m}^*\mathbf{H}_{n,m}\sum_{\ell \neq n }^N\mathbf{F}_{RF}\mathbf{f}^\ell_{BB}\sum_{q=1}^M\sqrt{P_{\ell,q}^\prime}s_{\ell,q}}_{\text{inter-cluster interference}} + \underbrace{\mathbf{w}^*_{n,m}\mathbf{n}_{n,m}}_{\text{noise}},
\end{align}
where $\mathbf{w}_{n,m}$ is of size $T_\text{MU}\times 1$. After combining, each MU decodes the intended signal by using successive interference cancellation (SIC)~\cite{higuchi2015non}. More details on SIC will be provided in Section III. For the sake of simplicity, $P_{n,m}$ denotes the normalized transmitted power such that $P_{n,m}=\frac{P_{n,m}^\prime}{\sigma^2}$.  

In mmWave communications, the single-path channel between the BS and  MU-$(n,m)$ can be expressed as
\begin{equation}\label{eq4}
    \mathbf{H}_{n,m} = \sqrt{T_\text{BS}T_\text{MU}} \beta_{n,m}\mathbf{a}_\text{MU}(\vartheta_{n,m})\mathbf{a}_\text{BS}^*(\varphi_{n,m}),
\end{equation}
where $\beta_{n,m} = g_{n,m}D_{n,m}^{\frac{-\nu}{2}}$, $g_{n,m}$ is the complex gain with zero-mean and unit-variance, $D_{n,m}$ is the distance between the BS and MU-$(n,m)$, and $\nu$ is the path loss factor~\cite{yang2017noma}. Also,  $\mathbf{a}_\text{BS}(\varphi_{n,m})$ and $\mathbf{a}_\text{MU}(\vartheta_{n,m})$ are the antenna array response vectors of the BS and MU-$(n,m)$, respectively, where $\vartheta_{n,m}$ and $\varphi_{n,m} \in [-1,1]$  are related to the angle of arrival (AoA) $\vartheta \in [-\frac{\pi}{2},\frac{\pi}{2}]$  and angle of departure (AoD)  $\phi \in [-\frac{\pi}{2},\frac{\pi}{2}]$ as $\vartheta_{n,m} = \frac{2d\text{sin}(\vartheta)}{\lambda}$ and $\varphi_{n,m} = \frac{2d\text{sin}(\phi)}{\lambda}$, respectively.  

In particular, for a uniform linear array (ULA), $\mathbf{a}_\text{BS}(\varphi_{n,m})$ is given by
\begin{equation}\label{eq5}
\mathbf{a}_{\text{BS}}(\varphi_{n,m}) = \frac{1}{\sqrt{T_\text{BS}}}\left[1, e^{-j\pi\varphi_{n,m}},\cdots, e^{-j\pi(T_\text{BS}-1)\varphi_{n,m}}\right]^\dagger,
\end{equation}
with $d$ is the antenna spacing and $\lambda$ denotes wavelength. The antenna array response vector for $\mathbf{a}_\text{MU}(\vartheta_{n,m})$ also has a similar structure to that of \eqref{eq5}~\cite{el2014spatially,alkhateeb2015limited}.  

To solely quantify the effect of digital/analog precoding on the sum-rate of HB-NOMA systems, throughout this paper several assumptions are considered: 
\begin{itemize}
    \item Full CSI of each user, $\mathbf{H}_{n,m}$, $m=1,2,\cdots,M$ and $n=1,2,\cdots,N$, is available at that MU. 
    \item The BS and all MUs steer the beams with continuous angles. That is, the quantization error is neglected for $\mathbf{F}_\text{RF}$ and $\mathbf{w}_{n,m}$, $m=1,2,\cdots,M$ and $n=1,2,\cdots,N$.
    \item The first MU of each cluster feeds complete effective channel back to the BS, i.e., infinite-resolution codebooks are used.
    \item The BS knows all MUs' channel gains $\left|\beta_{n,m}\right|$, for $m=1,2,\cdots,M$ and $n=1,2,\cdots,N$.  
    \item Each MU is capable of preforming error-free SIC. 
\end{itemize}

\section{Problem Formulation}\label{sec:problem}
As mentioned, the goal of this work is to quantify the impact of combining hybrid beamforming and NOMA on the sum rate of a dense network. In particular, we study the effect of jointly designing the hybrid precoders and analog combiners for clustered MUs on the sum-rate of the overall system.

In~(\ref{eq3}), after applying superposition coding at the transmitter, each user experiences two forms of interference. The intra-cluster inference which is due to other MUs within the cluster and inter-cluster interference which is due to MUs within other clusters. 

Suppressing the intra-cluster interference directly depends on efficient power allocation and deploying SIC. At the receiver side, each user performs SIC to decode the desired signal. To do this, the signal of users that have more power are decided and subtracted from the received signal. This process continues until the intended user decodes its signal. The remainder can be categorized as intra-cluster interference. Here, we assume that each user can perfectly perform SIC decoding. To mitigate the inter-cluster interference, the transmitter needs to design a proper beamforming matrix. In this paper, we adopt zero-forcing beamforming (ZFBF) which achieves a balance between implementation complexity and performance.      

Sum-rate has been categorically used to analyze the performance of NOMA. Consequently, we evaluate the sum-rate for the proposed HB-NOMA. The sum-rate is expressed as
\begin{equation}\label{sumrate}
    R_\text{sum} = \displaystyle \sum_{n=1}^N\sum_{m=1}^MR_{n,m},
\end{equation}
where $R_{n,m}$ is defined as
\begin{equation}\label{eq6}
    R_{n,m} = \text{log}_2\left(1 + \frac{P_{n,m}\left|\mathbf{w}_{n,m}^*\mathbf{H}_{n,m}\mathbf{F}_\text{RF}\mathbf{f}_\text{BB}^n\right|^2}{I_\text{intra}^{n,m} + I_\text{inter}^{n,m} + 1}\right).
\end{equation} 
Here, the intra-cluster interference after SIC processing at the $m$th user in the $n$th cluster, $I_\text{intra}^{n,m}$, is given by
\begin{equation}\label{eq61}
I_\text{intra}^{n,m} = 
\sum_{k=1}^{m-1}P_{n,k}\left|\mathbf{w}_{n,m}^*\mathbf{H}_{n,m}\mathbf{F}_\text{RF}\mathbf{f}_\text{BB}^n\right|^2. 
\end{equation}
Notice that the first user in each cluster is assumed to have the lowest allocated power. Further, the inter-cluster interference at the $m$th user in the $n$th cluster, $I_\text{inter}^{n,m}$, is given by 
\begin{equation}\label{eq62}
I_\text{inter}^{n,m} = \sum_{\ell\neq n}^{N}\sum_{q=1}^{M}P_{\ell,q}\left|\mathbf{w}_{n,m}^*\mathbf{H}_{n,m}\mathbf{F}_\text{RF}\mathbf{f}_\text{BB}^\ell\right|^2. 
\end{equation}

To improve the achievable rate, we need to obtain the optimal hybrid precoder $\breve{\mathbf{F}}_\text{RF}$, $\breve{\mathbf{f}}^{n}_\text{BB}$, for $n=1 \cdots N$, combiner $\breve{\mathbf{w}}_{n,m}$, for $n=1 \cdots N$,  $m=1 \cdots M$, and transmit power $\breve{P}_{n,m}$, for $n=1 \cdots N$,  $m=1 \cdots M$. To this end, the following optimization problem must be solved
\begin{IEEEeqnarray*}{rCl}
 &\underset{\{\mathbf{F}_\text{RF},\mathbf{f}_\text{BB}^n,\mathbf{w}_{n,m},P_{n,m}\}}{\text{maximize}} \ & \sum_{n=1}^N\sum_{m=1}^M R_{n,m}
 \IEEEyesnumber\\
&\text{subject to} & \left|\mathbf{F}_\text{RF}(i,j)\right|^2  = T_\text{BS}^{-1} \IEEEyessubnumber*\\
& &\norm[\big]{\mathbf{F}_\text{RF}[\mathbf{f}_\text{BB}^1, \mathbf{f}_\text{BB}^2,\cdots,\mathbf{f}_\text{BB}^N]}^2_F = N \\
& & \left|\mathbf{w}_{n,m}(i,j)\right|^2  = T_\text{MU}^{-1} \\
& & \sum_{n=1}^N\sum_{m=1}^M P_{n,m}\leq P_\text{t}, \\
& &  P_{n,m} > 0 
\end{IEEEeqnarray*}
where $i=1,2,\cdots,M$, $j=1,2,\cdots,M$, and $P_\text{t}$ equals to the signal-to-noise ratio (SNR), i.e, $P_\text{t} = \text{SNR}$. 
\section{The Maximization Algorithm}\label{algorithm}
Unfortunately, the maximization problem in (7) is non-convex since the objective function has a complicated expression and there is a coupling between the analog and digital precoders. Thus, finding a closed-form solution is non-trivial. Hence, we propose a simple but effective algorithm in three steps.

In the first step the BS and MU-$(n,m)$ solve the following problem
\begin{equation}\label{eq7}
    \underset{\{\mathbf{w}_{n,m},\mathbf{f}_\text{RF}^{n,m}\}}{\text{maximize}} \ \left|\mathbf{w}_{n,m}^*\mathbf{H}_{n,m}\mathbf{f}_\text{RF}^{n,m}\right| \qquad \text{subject to (10a) and (10c)}. 
\end{equation}
 Since the channel $\mathbf{H}_{n,m}$ has only one path, and given the continuous beamsteering capability assumption, the optimal solutions will be $\breve{\mathbf{w}}_{n,m} = \mathbf{a}_\text{MU}(\vartheta_{n,m})$ and $\breve{\mathbf{f}}_\text{RF}^{n,m} = \mathbf{a}_\text{BS}(\varphi_{n,m})$. To design the RF precoder, the BS selects the first user of each cluster. Here, the first MU is selected based on the following criterion,  
\begin{equation}\label{eq8}
    \left|\beta_{n,1}\right| > \left|\beta_{n,2}\right| > \cdots > \left|\beta_{n,M}\right| \quad \text{for} \quad n = 1, 2, \cdots, N.
\end{equation}
Stacking the the RF precoder vector of the first MUs obtained in~(\ref{eq7}), i.e., $\breve{\mathbf{f}}_\text{RF}^{n,1}$ for $n = 1, 2, \cdots, N$, we can find the RF or analog precoding matrix as
 \begin{equation}\label{eq81}
 \mathbf{F}_\text{RF} = \left[\breve{\mathbf{f}}^{1,1}_\text{RF}, \breve{\mathbf{f}}^{2,1}_\text{RF},\cdots,\breve{\mathbf{f}}^{N,1}_\text{RF}\right].     
 \end{equation}
The motivation for steering the beams to the first or closest user to the base station in each cluster is dictated by the structure of NOMA. More detail will be given in the next step. 
 
In the second step, the effective channel for MU-$(n,m)$ is expressed as
\begin{equation}\label{eq9}
    \overline{\mathbf{h}}_{n,m}^* =  \mathbf{w}_{n,m}^*\mathbf{H}_{n,m}\mathbf{F}_\text{RF} = \sqrt{T_\text{BS}T_\text{MU}}\beta_{n,m}\mathbf{a}_\text{BS}^*(\varphi_{n,m})\mathbf{F}_\text{RF}.
\end{equation}
We define the effective channel matrix as
\begin{equation}\label{eq91}
\overline{\mathbf{H}} = \left[\overline{\mathbf{h}}_{1,1}, \overline{\mathbf{h}}_{2,1}, \cdots, \overline{\mathbf{h}}_{N,1} \right]    
\end{equation}
where $\overline{\mathbf{h}}_{n,1}$ denotes the effective channel vector of MU-$(n,1)$. In HB-NOMA, the first MU of each cluster has to decode other MUs' signal in that cluster first before it can decode its own signal. Thus, by steering the beams toward the first MU within each cluster, we ensure that SIC can be performed in an effective fashion. Hence, the design of the RF precoder has to follow that of~(\ref{eq81}) to reduce the intra-cluster interference. Further, the effective channel matrix~(\ref{eq91}) causes the inter-cluster interference on the first MU in each cluster to be eliminated, which is outlined in more detail below. 
 
Recall that designing a proper digital or baseband precoder $\mathbf{F}_\text{BB}$ remarkably reduces the inter-cluster interference. Also, recall that we have selected zero-forcing method to design $\mathbf{F}_\text{BB}$. Thus, designing the baseband precoder equals to solving 
\begin{equation}\label{eq10}
    \underset{\{\mathbf{f}_\text{BB}^\ell\}_{\ell\neq n}}{\text{minimize}} \ \left|I_\text{inter}^{n,m}\right| \qquad \text{subject to (10b)}.
\end{equation}
Based on ZFBF, the solution for (\ref{eq10}) is obtained as~\cite{alkhateeb2015limited}
\begin{equation}\label{eq11}
    \mathbf{F}_\text{BB} = \overline{\mathbf{H}}^*\left(\overline{\mathbf{H}}\overline{\mathbf{H}}^*\right)^{-1}\bf{\Lambda},
\end{equation}
where the diagonal elements of $\mathbf{\Lambda}$ are given by~\cite{alkhateeb2015limited}
\begin{equation}\label{eq12}
    \mathbf{\Lambda}_{n,n} = \sqrt{\frac{T_\text{BS}T_\text{MU}}{\left(\mathbf{F}_\text{RF}^*\mathbf{F}_\text{RF}\right)_{n,n}^{-1}}}\left|\beta_{n,1}\right|, \quad \text{for} \quad n = 1, 2, \cdots, N.
\end{equation}
According to ~(\ref{eq11}) and the hybrid beamforming results in~\cite{alkhateeb2015limited}, we can conclude that the inter-cluster interference amongst the first MUs is zero, i.e., $\overline{\mathbf{h}}^*_{n,1}\mathbf{f}^\ell_\text{BB} = 0$ for $n = 1, 2, \cdots, N$ and $\ell \neq n$. That is to say, inter-cluster interference is perfectly eliminated for the first MUs. This completes our justification about the orienting the beams toward the first MUs and choosing their effective channel vector in designing $\mathbf{F}_\text{BB}$.   

At the third step, the BS first reorders the MUs, then allocates the power. The reordering process is done based on the effective channel vectors as 
\begin{equation}\label{eq121}
    \norm[\big]{\overline{\mathbf{h}}_{n,1}} > \norm[\big]{\overline{\mathbf{h}}_{n,2}}>, \cdots, 
    > \norm[\big]{\overline{\mathbf{h}}_{n,M}} \quad \text{for} \quad n = 1, 2, \cdots, N.
\end{equation}
It is worth mentioning that in~(\ref{eq8}) we aimed to find the first MUs based on the large-scale channel gain. However, in HB-NOMA the power allocation is conducted based on order of the effective channel gains, which takes both the large scale gain and the impact of hybrid beamforming and combining into account. Here, in~\eqref{eq121}, we propose to improve HB-NOMA's sum-rate by carrying out the power allocation based on the effective channel. Accordingly, it is true that while using~(\ref{eq8}) or~(\ref{eq121}) we arrive at the same first MU in each cluster, the order of the remaining MUs might be different when using~(\ref{eq121}) compared to~(\ref{eq8}). This can be intuitively explained by considering that the alignment of the MUs in each cluster with regard to the first MU may play a more important role on the quality of the their channels than their proximity to the base station. 

The optimal power allocation in~(10) is non-trivial and iterative procedures are needed to solve for the power of the $m$th user in the $n$th cluster according to
\begin{equation}\label{eq13}
    \underset{\{P_{n,m}\}}{\text{maximize}} \ \sum_{n=1}^{N}\sum_{m=1}^M R_{n,m}   \qquad \text{subject to (10d) and (10e)}. 
\end{equation}
Obtaining the optimal solution for (\ref{eq13}) is beyond the scope of this paper. Hence, we propose a suboptimal solution. Our solution has two stages. First the BS divides power equally amongst the clusters, i.e., $P_\text{c} = P_\text{t}/N$. Then a fixed power allocation~\cite{higuchi2015non} is utilized for the users in each cluster with regard to the constraint $\sum_{m=1}^M P_{n,m} = P_\text{c}$. 
\section{The Rate Evaluation for HB-NOMA}\label{sec:lower}
 In this section we concentrate on studying the achievable rate of an HB-NOMA MU.

 In practical scenarios, correlation amongst the users in each cluster is always imperfect which gives $\mathbf{a}_\text{BS}(\varphi_{n,1}) \neq \mathbf{a}_\text{BS}(\varphi_{n,2}) \neq \cdots \neq \mathbf{a}_\text{BS}(\varphi_{n,M})$ for $n = 1, 2, \cdots, N$. The reason is that since $\varphi_{n,k}$ and $\varphi_{n,m}$ for $k=m$ are independent, the probability for the event $\varphi_{n,k} = \varphi_{n,m}$ is zero~\cite{alkhateeb2015limited}. Of interest, we here study the impact of imperfect correlation on the sum-rate in this subsection. Before this, we calculate the norm of the effective channel defined in Eq.~(\ref{eq9}). Defining
 \begin{equation}\label{eqFejer}
 \left|\mathbf{a}_\text{BS}^*(\varphi_{n,m})\mathbf{a}_\text{BS}(\varphi_{\ell,1})\right|^2 = K_{T_\text{BS}}(\varphi_{\ell,1}-\varphi_{n,m}) 
 \end{equation}
 where $K_{T_\text{BS}}$ is Fej$\acute{\text{e}}$r kernel of order $T_\text{BS}$~\cite{strichartz2000way}, we get
 \begin{equation}\label{eq1601}
 \norm[\big]{\overline{\mathbf{h}}^\text{Im}_{n,m}}^2 = T_\text{BS}T_\text{MU}\left|\beta_{n,m}\right|^2\displaystyle \sum_{\ell = 1}^N K_{T_\text{BS}}\left(\varphi_{\ell,1}-\varphi_{n,m}\right).
 \end{equation} 

The following lemma describes the relationship between the effective channel vectors of MU-$(n,1)$ and MU-$(n,m)$ for $m=2,3,\cdots,M$ and $n=1,2,\cdots,N$.
 \begin{lemma}\label{lemma:2}
 The relationship between the imperfect effective channel for MU-$(n,m)$ and MU-$(n,1)$ for $m = 2, 3,\cdots, M$ and $n = 1,2,\cdots,N$ can be modelded as
 \begin{equation} \label{eq19}
    \widetilde{\mathbf{h}}_{n,m}^\text{Im} \approx \rho_{n,m}\widetilde{\mathbf{h}}_{n,1}^\text{Im} + \sqrt{1 - \rho_{n,m}^2}\boldsymbol{\varpi}_{n,m}
\end{equation}
where $\widetilde{\mathbf{h}}_{n,m}^\text{Im}$ denotes the normalized effective channel, $\rho_{n,m} = \left|\widetilde{\mathbf{h}}_{n,m}^\text{Im*}\widetilde{\mathbf{h}}_{n,1}^\text{Im}\right|$
, and $\boldsymbol{\varpi}_{n,m}$ is a unit-norm vector.
 \end{lemma}
 \begin{proof}
Assume that the effective channel vectors are fed back by using infinite-resolution codebook. Also, let $\widetilde{\mathbf{h}}_{n,m}^\text{Im}$ denotes the normalized effective channel vector for MU-$(n,m)$ given by
\begin{equation}\label{eq40}
\widetilde{\mathbf{h}}_{n,m}^\text{Im} = \frac{\overline{\mathbf{h}}_{n,m}^\text{Im}}{\norm[\big]{\overline{\mathbf{h}}^\text{Im}_{n,m}}}. \end{equation} 
For two complex-valued vectors $\overline{\mathbf{h}}_{n,m}^\text{Im},\overline{\mathbf{h}}_{n,1}^\text{Im} \in V_\mathbb{C}$ the angle between them is obtained as~\cite{scharnhorst2001angles}
\begin{equation}\label{anglevector}
\rho e^{j\omega} = \widetilde{\mathbf{h}}_{n,m}^\text{Im*}\widetilde{\mathbf{h}}_{n,1}^\text{Im},     
\end{equation}
where $(\rho\leq 1)$ is equal to~\cite{scharnhorst2001angles}
\begin{equation}\label{rho}
    \rho = \text{cos}{\Phi}_\text{H}(\widetilde{\mathbf{h}}_{n,m}^\text{Im},\widetilde{\mathbf{h}}_{n,1}^\text{Im}) = \left|\widetilde{\mathbf{h}}_{n,m}^\text{Im*}\widetilde{\mathbf{h}}_{n,1}^\text{Im}\right|,
\end{equation}
in which $\Phi_\text{H}(\widetilde{\mathbf{h}}_{n,m}^\text{Im},\widetilde{\mathbf{h}}_{n,1}^\text{Im})$, $0\leq \nu_\text{H}\leq\frac{\pi}{2}$, is called the Hermitian angle between two complex-valued vectors $\overline{\mathbf{h}}_{n,m}^\text{Im},\overline{\mathbf{h}}_{n,1}^\text{Im}$ and $\omega$, $-\pi\leq\omega\leq\pi$, is called their pseudo-angle. The factor $\rho$ is related to the angle between two lines in the complex vector space $V_\mathbb{C}$ while the angle $\omega$ is defined in the context of pseudoconformal transformations~\cite{scharnhorst2001angles,goldman1999complex}. Here, we rename the factor $\rho$ as the correlation between two vectors $\overline{\mathbf{h}}_{n,m}^\text{Im},\overline{\mathbf{h}}_{n,1}^\text{Im}$ and disregard the pseudo angle $\omega$. Therefore, the angle between the two vectors is approximated by $\rho$. Then, for some $n$ the two normalized vectors, $\widetilde{\mathbf{h}}_{n,m}^\text{Im}$ for $m = 2, 3, \cdots, M$ and $\widetilde{\mathbf{h}}_{n,1}^\text{Im}$, are related together through the factor $\rho_{n,m}$ and vector $\boldsymbol{\varpi}_{n,m}$ which is  a unit-norm in the null-space of $\widetilde{\mathbf{h}}_{n,1}^\text{Im}$.  
 \end{proof}

 Now, we derive a lower bound for the sum-rate of MU-$(n,m)$ in HB-NOMA systems when the correlation is imperfect. 
\begin{theorem}\label{theo:2}
Regarding the imperfect correlation, the lower bound of the rate of HB-NOMA MU-$(n,m)$ for $m>1$ cluster is given by
\begin{equation}\label{eq16}
    R_{n,m}^\text{Im} \geq \text{log}_2\left(1 +\frac{P_{n,m}\rho_{n,m}^2T_\text{BS}T_\text{MU}\left|\beta_{n,m}\right|^2}{\zeta_\text{intra}^{n,m}+ \zeta_\text{inter}^{n,m} + \zeta_\text{noise}^{n,m}}\right)
\end{equation}
where 
\begin{equation}\label{eq161}
\zeta_\text{intra}^{n,m} = \sum_{k = 1}^{m-1}P_{n,k}\rho_{n,m}^2T_\text{BS}T_\text{MU}\left|\beta_{n,m}\right|^2,
\end{equation}
and
\begin{align}\label{eq162}
\zeta_\text{inter}^{n,m} = &  P_\text{c}\left(1-\rho_{n,m}^2\right)T_\text{BS}T_\text{MU}\left|\beta_{n,m}\right|^2\lambda_\text{max}\left(\mathbf{S}\right)\eta_nK_{T_\text{BS},\Sigma_1},
\end{align}
$\lambda_\text{max}\left(\mathbf{S}\right)$ showing the maximum eigenvalue of $\mathbf{S} = \mathbf{F}_\text{BB}^{-n}\mathbf{F}_\text{BB}^{-n*}$.  $\mathbf{F}_\text{BB}^{-n}$ denotes the $\mathbf{F}_\text{BB}$ after eliminating the $n$th column. Also, for some $m$ we define
\begin{equation}\label{eq163}
K^{-1}_{T_{BS},\Sigma_m} = \left(\displaystyle \sum_{\ell = 1}^N K_{T_\text{BS}}\left(\varphi_{\ell,1}-\varphi_{n,m}\right)\right)^{-1}    
\end{equation}
 with $K_{T_\text{BS}}\left(\varphi_{\ell,1}-\varphi_{n,m}\right)$ denotes the Fej$\acute{\text{e}}$r kernel in~(\ref{eqFejer}). Finally $\zeta_\text{noise}^{n,m}$ is expressed as
\begin{equation}\label{eq164}
    \zeta_\text{noise}^{n,m} = \eta_nK_{T_{BS},\Sigma_1}K^{-1}_{T_{BS},\Sigma_m},
\end{equation}
where $K^{-1}_{T_{BS},\Sigma_m}$ is defined in~(\ref{eq163}) when $\varphi_{n,1}$ is replaced by $\varphi_{n,m}$.
 \end{theorem}
The proof is eliminated due to lack of enough space. Since for MU-$(n,1)$ the correlation factor $\rho = 1$ even in the case of imperfect correlation, Thoerem 1 is still valid for these users.    
Theorem~\ref{theo:2} clearly states that the achievable rate of each MU extremely depends on the correlation between the first and the intended MU such as a weak correlation reduces the power of the effective channel of that MU. Besides, the bound gives an useful insight on the clustering the MUs. That is, the relation between the AoD $\phi_{n,m}$ and the correlation factor $\rho_{n,m}$ is nonlinear. On the other hand, $\phi_{n,m}$ does not explicitly appears in the rate expression whereas $\rho_{n,m}$ does. As a result, the clustering can be done through $\rho_{n,m}$s as a design criterion.

The following two remarks are in order:
\begin{itemize}
    \item This paper concentrates only on studying the rate performance of HB-NOMA systems under imperfect effective channel correlation, several ideal assumptions are considered in Section~\ref{sec:system}. Studying the performance of the imperfect effective channel correlation along with practical assumptions can be done in the future works. For instance, in practical scenarios the effective channel at the transmitter is quantized. The impact of the quantization error on the sum-rate of proposed HB-NOMA system with imperfect correlation is required.
    \item In the algorithm proposed in Section~\ref{algorithm} the fixed intra-cluster and inter-cluster power allocation strategies are ordered. Adaptive but complex power allocation procedure for MUs inside each cluster can be designed through iterative algorithm~\cite{zhang2016robust}. On the other hand, in practice, each cluster poses different inter-cluster interference. Therefore, the algorithm will be suboptimum in term of inter-cluster power allocation. A useful inter-cluster power allocation regarding the interference caused by each cluster is proposed in~\cite{chen2017exploiting}.  
\end{itemize}
 \section{Numerical Results}\label{sec:simulation}
 In this section we discuss the numerical simulations for the achievable rate of the HB-NOMA users  under the imperfect correlation assumption.      

\begin{figure}
     \centering
     \includegraphics[scale=0.4]{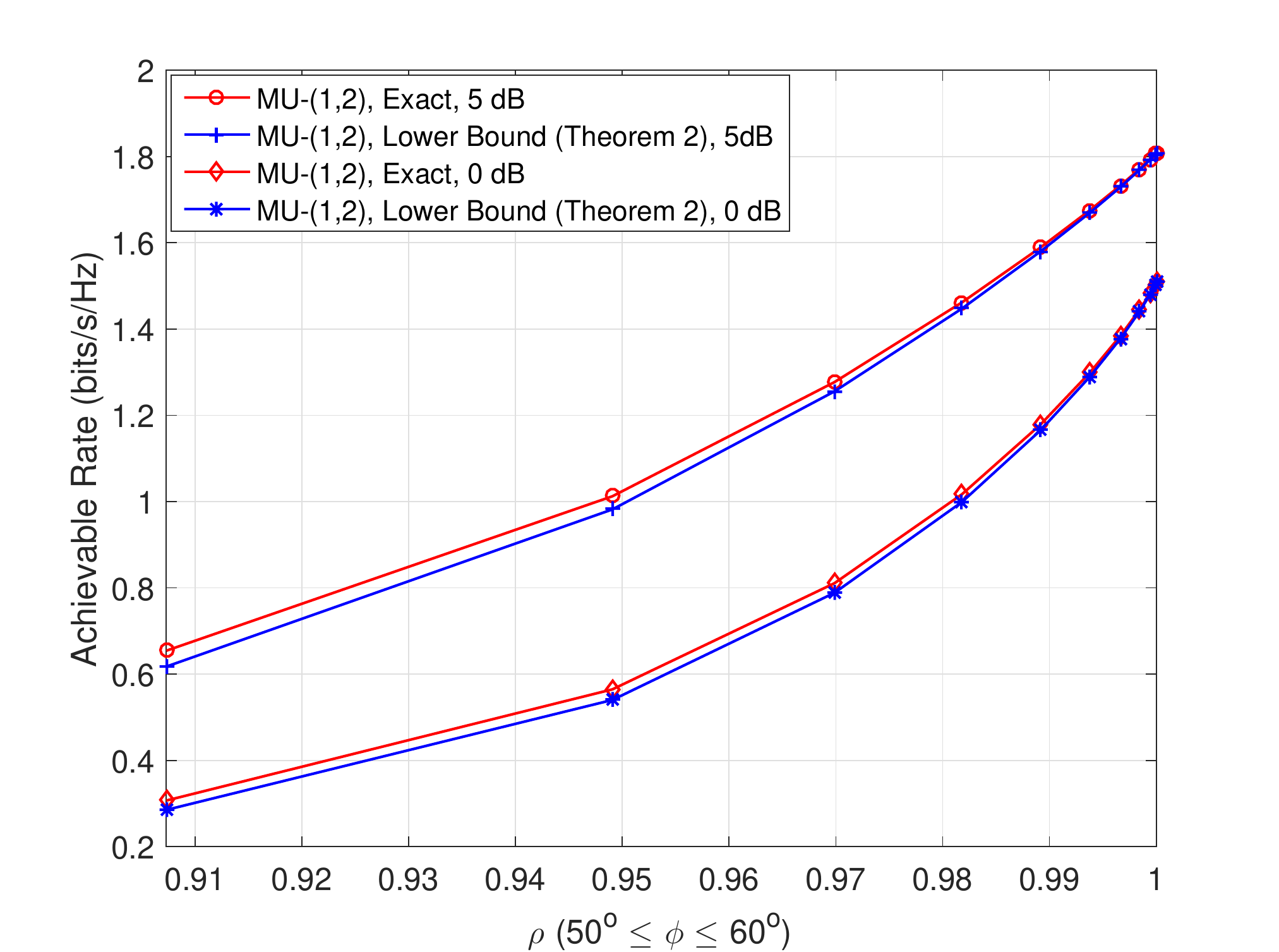}
     \caption{The achievable rate vs correlation factor $\rho$ performance for the imperfect correlation between MU-(1,1) and MU-(1,2) for SNR 5 and 0 dB.}
     \label{fig:corr}
 \end{figure}
 
\begin{figure}
     \centering
     \includegraphics[scale=0.4]{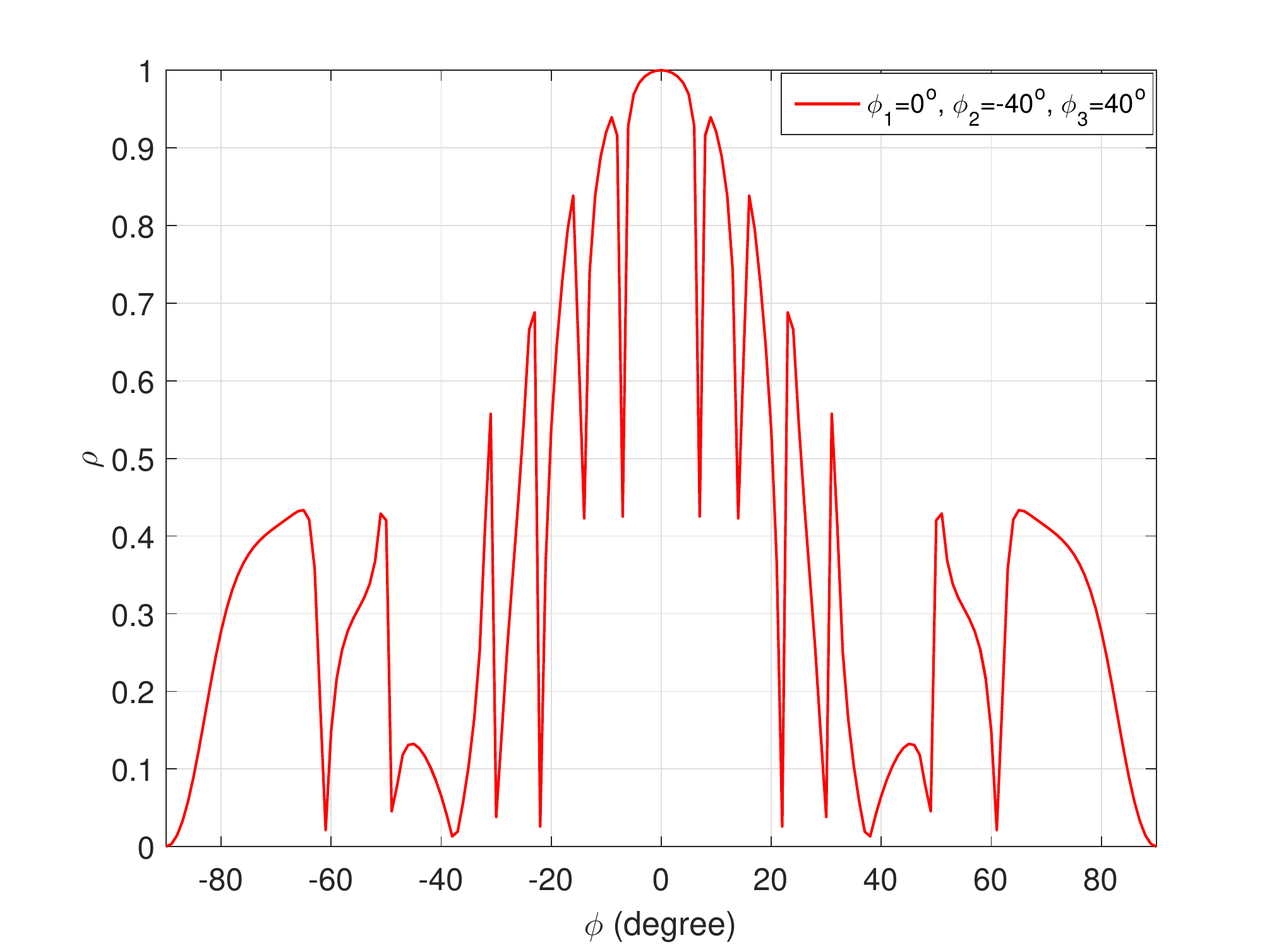}
     \caption{Relationship between the correlation and AoD for MU-$(1,2)$ when there are three clusters. The AoDs for MU-$(n,1)$ for $n=1, 2, \text{and}, 3$ are assumed to be respectively $0^\circ$, $-40^\circ$, and $40^\circ$.} 
     \label{fig:1RhoVsAngle}
 \end{figure}

Fig.~\ref{fig:corr} evaluates the effect of the imperfect correlation on the rate performance of MU-(1,2). We consider that the BS is equipped with a $16\times 1$ ULA which serves two clusters each containing two MUs. Each MU has a $4\times 1$ ULA. Also, there is a single path channel from the BS to each MU and full CSI is available at the BS. The elevation AoAs/AoDs have a uniform distribution in $\left[-\frac{\pi}{2},\frac{\pi}{2}\right]$. Further, the proposed algorithm in Section~\ref{algorithm} is used to maximize the rate where The allocated power for the close MU is assumed to be 1/4$P_\text{c}$ and for the far MU it is 3/4$P_\text{c}$ with $P_\text{c} = 1/2P_\text{t}$. Finally, the large-scale fading, i.e, $D^{-\nu}$ is assumed 0 dB for MU-(1,1) and MU-(2,1) and -10 dB for MU-(1,2) and MU-(2,2). It is assume that the SNR is 0 and 5 dB. Also, the AoD, i.e., $\phi_{1,2}$, is increasing from $50^\circ$ to $60^\circ$. As $\phi_{1,2}$ grows up, the correlation factor increases but it has a nonlinear behavior. This means $\rho$ is a nonlinear function of AoDs. Obviously, the small correlation considerably  degrades the rate performance, e.g., a channel correlation 0.92 decreases the rate about 1 bits/s/Hz compare to $\rho = 1$. By increasing the correlation factor the rate increases. That is, by increasing $\rho$ the effect of inter-cluster,  i.e, $\zeta_\text{inter}$, decreases and angle of the effective channel of the MU approaches to that of MU-(1,1) which leads to a higher value for $K_{T_\text{BS}}$ in Theorem~2. Also, the figure shows that the derived lower bound in Theorem~\ref{theo:2} is accurate and close to the simulation value.

In Fig.~\ref{fig:1RhoVsAngle} we investigate the relationship between the  correlation and AoD for $-90^\circ \leq \phi \leq 90^\circ$  in the case of three clusters and two MUs in each of them. First we set AoD of MU-$(n,1)$ for $n=1, 2, 3$ respectively $0^\circ$, $-40^\circ$, and $40^\circ$. As red dash line shows the correlation between MU-$(1,1)$ and MU-$(1,2)$ around $0^\circ$ is 1. It also shows that when AoD of MU-$(1,2)$ is in range $[-7^\circ,7^\circ]$ the correlation remains greater than $0.95$.  The lowest correlation happens around $-40^\circ$ and $40^\circ$ which are the AoDs of the second and third clusters. 

 \section{Conclusion}\label{sec:conclusion}
Hybrid beamforming-based NOMA has been studied. To this end, we formulated an optimization problem for the sum-rate of HB-NOMA system. Then, due to the complicate objective function and constrains an algorithm is proposed in three steps. In order to evaluate the sum-rate, we derived a lower bound for each MU for imperfect correlation between the effective channel of the first MU and other MUs. Our lower bound analysis demonstrates that under the assumption of imperfect correlation the lower bound indicates that an improper correlation can cause a remarkable degradation in the rate performance. The numerical results support our analytical findings.

\appendices




\ifCLASSOPTIONcaptionsoff
  \newpage
\fi



%
\bibliographystyle{IEEEtran}
\bibliography{IEEEabrv,references}

%







\end{document}